%% file: WPT_ArXiv.tex
\documentclass[12pt,onecolumn,draftclsnofoot]{IEEEtran} 
\usepackage{fancyhdr}
\usepackage{epsfig}
\usepackage{threeparttable}
\usepackage{epsf,epsfig}
\usepackage{amsthm}
\usepackage{amsmath}
\usepackage{amssymb}
\usepackage{amsfonts}
\usepackage[noadjust]{cite}
\usepackage{dsfont}
\usepackage{subfigure}
\usepackage{color}
\usepackage{booktabs}

\input{header}

\begin{document}
\title{\huge \setlength{\baselineskip}{30pt} Enabling Wireless  Power Transfer in Cellular Networks: Architecture, Modeling  and Deployment}
\author{Kaibin Huang and Vincent K. N. Lau\thanks{\setlength{\baselineskip}{15pt} K. Huang is with the Hong Kong Polytechnic University, Hong Kong and V. K. N.  Lau  is with the Hong Kong University of Science and Technology, Hong Kong.  Email: huangkb@ieee.org, eeknlau@ece.ust.hk.  Updated on \today.  }}

\maketitle

\begin{abstract}
Microwave  power transfer (MPT)  delivers energy wirelessly from stations called \emph{power beacons} (PBs) to mobile devices by  microwave radiation. This provides mobiles  practically infinite battery lives and eliminates the need of power cords and chargers. To enable MPT for mobile charging, this paper proposes a new network architecture that overlays an uplink  cellular network with randomly deployed PBs for powering mobiles, called a \emph{hybrid network}. The deployment of the hybrid network  under an outage constraint on data links is investigated based on a stochastic-geometry model where single-antenna base stations (BSs) and PBs form independent homogeneous  Poisson point processes (PPPs) with densities $\lambda_b$ and $\lambda_p$, respectively, and   single-antenna mobiles are uniformly distributed in Voronoi cells generated by BSs.   In this model, mobiles and PBs fix their transmission power  at $p$ and $q$, respectively; a PB either radiates isotropically, called \emph{isotropic MPT}, or directs energy towards target mobiles by beamforming, called \emph{directed MPT}. The model is applied to derive the tradeoffs between the network parameters $(p, \lambda_b, q, \lambda_p)$ under the outage constraint.  First, consider the deployment of the cellular network. It is proved that the outage constraint is satisfied so  long as the product $p\lambda_b^{\frac{\alpha}{2}}$ is above a given threshold where $\alpha$ is the path-loss exponent. Next, consider the deployment of the hybrid network assuming infinite energy storage at mobiles. It is shown that for isotropic MPT, the product $q\lambda_p \lambda_b^{\frac{\alpha}{2}}$ has to be above a given threshold so that PBs are sufficiently dense; for directed MPT, $z_mq\lambda_p \lambda_b^{\frac{\alpha}{2}}$ with $z_m$ denoting the array gain  should exceed a different threshold to ensure short distances between PBs and their target mobiles. Furthermore, for directed MPT, $(z_mq)^{\frac{2}{\alpha}}\lambda_b$ has to be sufficiently large as otherwise PBs fail to  deliver sufficient power to target mobiles regardless of power-transfer  distances. Last, similar results are derived for the case of mobiles having small energy storage. 
\end{abstract}


\section{Introduction}

One of the most desirable new features for mobile devices is wireless charging  that eliminates the need of power cords and chargers. To realize this feature, the paper proposes a novel network architecture (see Fig.~\ref{Fig:Network}) where stations called \emph{power beacons} (PBs) are deployed in an existing cellular network for charging  mobiles via microwave radiation known as \emph{microwave power transfer} (MPT). Charging a typical smartphone requires average received microwave power (MP) of tens to hundreds of  milliwatts \cite{TecReview:WirelessPowerCellPhone:2009}. For instance, a conventional smartphone powered by a  battery supplying $1000$ mAh at $3.7$ V and a battery life of $24$ hr can be self-sustaining by WPT if the average transferred  power is above $154$ mW assuming lossless   RF-DC conversion. The challenges in realizing wireless charging by MPT are threefold, namely creating line-of-sight (LOS) links from PBs to mobiles to enable close-to-free-space power transfer, forming sharp energy beams at PBs to counteract propagation loss, and reducing power consumption of mobiles.  These challenges can be tackled by 
the advancements of three technologies, respectively. First, PBs require neither  backhaul links nor complex computation, allowing low-cost and dense deployment.  Specifically, PBs can be installed wherever there are connections to the electrical  grid or even in vehicles and airplanes. Deploying many PBs shortens transfer distances, resulting in LOS power-transfer links. Second, massive antenna arrays with tens to hundreds of elements are a technology currently under active development \cite{Marzetta:CellularUnlimitedBSAntennas:2010} and can be used at PBs for forming sharp energy beams towards mobiles to achieve the WPT efficiency close to one. As discussed in the sequel, gigantic  antenna arrays have been used to achieve an efficiency of $45\%$ for MPT over $36,000$ km  (from a solar-power satellite to the earth). Third, cellular networks are evolving towards the deployment of dense base stations that will substantially reduce the required transmission power of mobiles \cite{ChanAndrews:FemtocellSurvey:2008}. Consequently, PBs' transmission power can be reduced or the MPT distances increased. Furthermore, the advancements in the areas of lower power electronics and displays (e.g., reflective displays) will also decrease  the power consumption of mobile devices and increase  the practicality of MPT.

In this paper,  an uplink cellular  network overlaid  with PBs, called the \emph{hybrid network}, is studied using a stochastic-geometry model where base stations (BSs) and PBs form  independent homogeneous Poisson point processes (PPPs) and 
mobiles are uniformly distributed in  corresponding Voronoi cells with respect to BSs. In the model, PBs are deployed for powering uplink transmissions under an outage constraint by either isotropic radiation or beamforming towards target mobiles, called \emph{isotropic} or \emph{directed} MPT, respectively. The model is used to derive  the tradeoffs between the network parameters including the PB/mobile transmission power and PB/BS densities  for different network configurations accounting for isotropic/directed MPT and mobiles having large/small energy storage. The results provide insight into the hybrid-network deployment and throughput.  

\begin{figure}[t]
\begin{center}
\includegraphics[width=15cm]{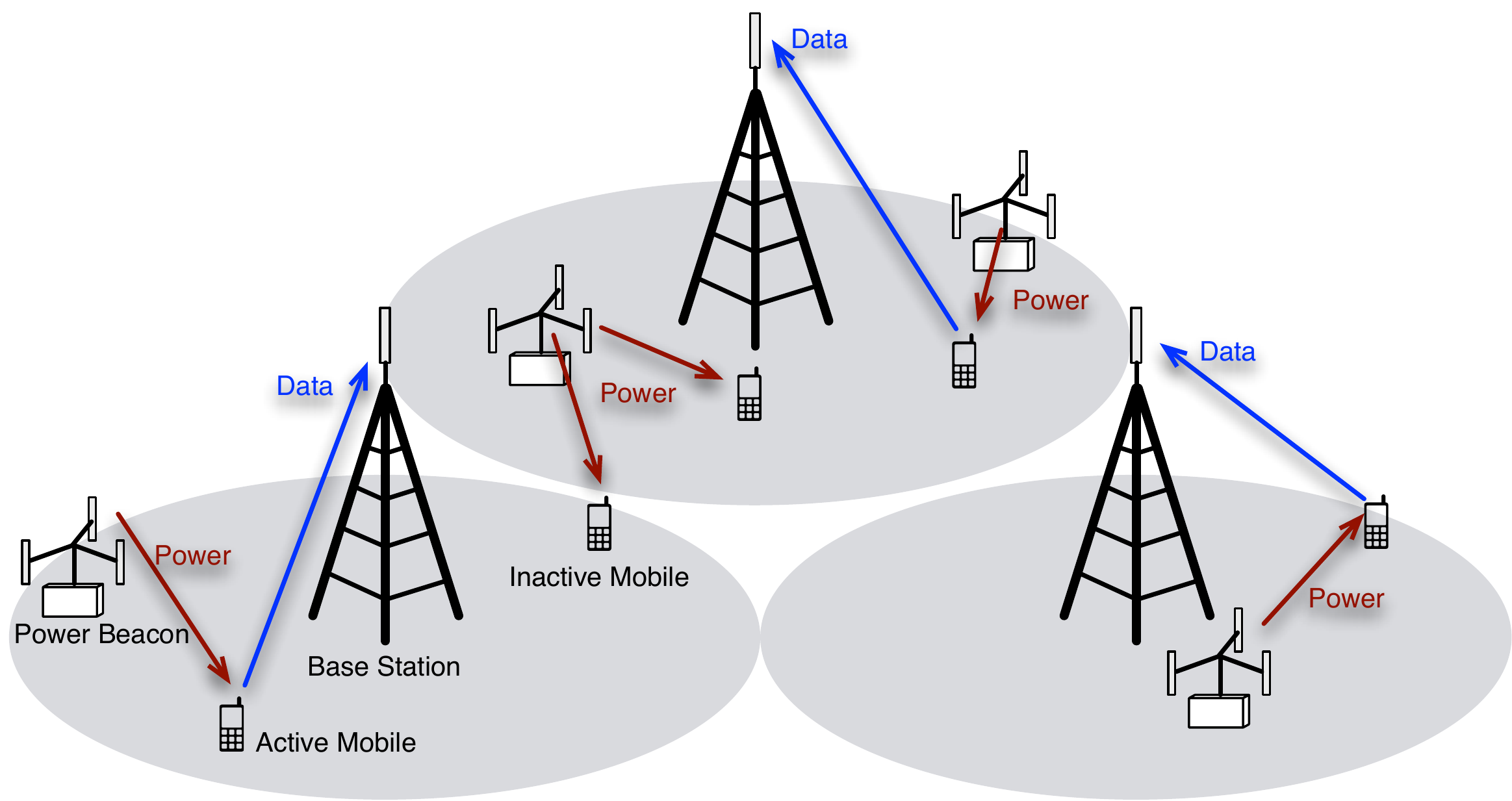}
\end{center}
\caption{Hybrid network that overlays an uplink cellular network with randomly deployed power beacons that wirelessly charge mobiles via microwave radiation.}
\label{Fig:Network}
\end{figure}

\subsection{Microwave Power Transfer} 

Wireless power transfer  is  usually implemented in practice based on one of the three  different technologies, namely  inductive coupling \cite{Want:introductionRFIDTech:2006}, magnetic resonance coupling \cite{Kurs:WPTMagneticResonances:2007},  and MPT \cite{Brown:RadioWPTHistory:1984} corresponding to short range (tens of centi-meters), mid range (several meters), and long range (up to tens of kilometers), respectively. Besides limited transfer distances, both inductive and magnetic-resonance coupling require  alignment and calibration of coils at transmitters and receivers, making these technologies unsuitable for mobile charging. MPT has no such limitations but suffers from potentially severe propagation loss over a long transfer  distance.  

Such loss can be overcome by free-space beaming and intercepting  microwave  energy  using antennas with large apertures. Based on this principle,  efficient point-to-point MPT  over long distances has been demonstrated in large-scale systems. In 1964, W. C. Brown first showed  the practicality  of MPT by demonstrating   a microwave-powered small helicopter hovering  at an altitude of $50$-feet, where a rectenna was used to intercept power beamed from the ground using the frequency of $2.45$ GHz and supply  DC power of $270$ W \cite{Brown:ExperimentMicrowaveBeamHelicopter}. This experiment motivated researchers in Canada and Japan  to develop microwave-powered airplanes in $1980$s \cite{Schlesak:MicrowavePoweredHiAltitudePlatform:1998} and  $1990$s \cite{Matsumoto:SolarPowerSatJapan:2002}, respectively. An extreme MPT system  was designed for the National Aeronautics and Space Administration (NASA) of the United States for transferring  power of about $2.7$ GW  from a solar-power satellite to a station on the surface of the earth \cite{Mcspadden:SpaceSolarPowerMicrowaveWPT:2002}. The power-transfer efficiency is predicted to be $45\%$ over a transfer distance as long as $36,000$ km by using gigantic transmit and receive antenna arrays with diameters of $500$ m and $7.5$ km, respectively. 
Large antennas with the dimensions  in the aforementioned experiments are impractical and unnecessary for  MPT in the hybrid network where efficient MPT relies  on deploying dense PBs to shorten power-transfer distances. 

Recently, simultaneous transfer of  information and power has emerged to be a new research area, which traditionally was  treated as two separate problems \cite{Varshney:TransportInformationEnergy:2008, GroverSahai:ShannonTeslaWlssInfoPowerTransfer,
ZhouZhang:WlessInfoPowrTransfer:RateEnergy:2012, Zhang:MIMOBCWirelessInfoPowerTransfer}. In \cite{Varshney:TransportInformationEnergy:2008}, the author investigated the information capacity of a point-to-point wireless channel under the constraint that the average received  power  is above a given threshold, establishing the relation between the transferred power and information rate over the same channel. For the channel in an inductive-coupling circuit, the efficiency   of power/information transfer  increases/decreases with the input-signal bandwidth. Note that the power-transfer efficiency is the highest at the resonance frequency of the circuit and decreases as the transfer frequency deviates from the resonance point.  This results in a  tradeoff between the transferred power and information capacity that is characterized in \cite{GroverSahai:ShannonTeslaWlssInfoPowerTransfer}. A two-user broadcast channel is considered in \cite{Zhang:MIMOBCWirelessInfoPowerTransfer} where a multi-antenna base station transmits information to one multi-antenna terminal and energy to the other. A tradeoff between the power and information rate delivered over the broadcast channel is achieved   by precoding at the base station. Different from prior work, this paper considers coupled wireless power and information transfer in a large-scale  network and the resultant  relation between the hybrid-network parameters  can be interpreted as the network counterparts of the existing  capacity-and-power tradeoffs for small-scale systems.

\subsection{Modeling the Access and Power-Beacon Networks}
Cellular networks traditionally are modeled using hexagonal grids but such models lack tractability and fail to account for ad hoc BS deployment that is a trend of network evolution. These drawbacks can be overcome by using stochastic geometry for modeling the cellular-network architecture \cite{FossZuyev:VoronoiProcessPoisson:1996, Baccelli:StochGeometryArchitectCommNetwork:2006, Andrews:TractableApproachCoverageCellular:2010, NovlanAndrews:ModelUplinkCellular:2012}. In \cite{FossZuyev:VoronoiProcessPoisson:1996}, a cellular network was modeled as a bivariate PPP that comprises two independent PPPs representing BSs and mobiles, respectively. As a result,  the network architecture is  distributed as a Poisson Voronoi tessellation. This random network model was applied in  \cite{Baccelli:StochGeometryArchitectCommNetwork:2006} to study the economics of cellular networks, and combined with a channel-fading model  to study the  downlink coverage  in \cite{Andrews:TractableApproachCoverageCellular:2010} and uplink coverage in \cite{NovlanAndrews:ModelUplinkCellular:2012}, yielding analytical  results consistent with the performance of practical networks. An uplink network is considered in this paper where single-antenna mobiles transmit data to single-antenna BSs under  an outage-probability constraint for a target received signal-to-interference-and-noise ratio (SINR).  Aligned with existing approaches,  BSs in the uplink network is modeled as a homogeneous PPP with density $\lambda_b$ and each cell of the resultant Voronoi tessellation comprises a uniformly distributed active mobile. Data links have path loss with the exponent $\alpha > 2$  but  no fading.  

Mobiles  in the current cellular-network model are powered by MPT rather than reliable power supplies assumed in prior work.  Randomly deployed PBs are modeled as a homogeneous PPP similar to the existing models of 
wireless ad hoc networks (see e.g., \cite{HaenggiAndrews:StochasticGeometryRandomGraphWirelessNetworks}). 
MPT  uses frequencies outside the data bandwidth (e.g., in the ISM band) and hence causes no interference to uplink transmissions.  Given fixed PB transmission power, the MPT bandwidth is required to be  sufficiently large so that the power-spectrum density meets regulations on  microwave radiation. 
Like data links, MPT links have no fading and  their path-loss exponent  $\beta > 2$ need not be equal to $\alpha$ due to the difference in  transmission ranges and frequencies. Each PB radiates energy either isotropically or along the directions of targeted mobiles by beamforming, referred to as \emph{isotropic} and \emph{directed} MPT, respectively. Mobiles intercept energy transferred from PBs continuously and store it for powering subsequent uplink transmission.

\subsection{Contributions and Organization}

The main contributions of   the  paper are summarized as follows. Let $\{c_n\}$ denote a set of constants to be derived in the sequel. 
\begin{enumerate}
\item First, consider the deployment of the cellular network. Define the \emph{feasibility region} $\mathcal{F}_c$ as all feasible combinations of the network parameters $(p, \lambda_b)$ under the outage constraint. It is proved that 
\begin{equation}\label{Intro:Feasibility:a}
\mathcal{F}_c = \l\{(p, \lambda_b) \in \mathds{R}_+^2\mid  p\lambda_b^{\frac{\alpha}{2}}\geq c_1\r\}
\end{equation}
where $\mathds{R}_+$ denotes the set of nonnegative numbers. It is observed that  the minimum $p$ increases with decreasing $\lambda_b$ and vice versa. 
Consider the special case of an interference-limited network with  noise omitted and under an additional constraint that requires received signal power at each BS to be above a threshold for given probability. 
It is shown that outage-probability is independent with $p$ and $\lambda_b$, making  the outage constraint  irrelevant.  The corresponding feasibility region $\mathcal{F}_c$ is found to have a similar form as in \eqref{Intro:Feasibility:a}:
\begin{equation}\nn
\mathcal{F}_c = \l\{(p, \lambda_b) \in \mathds{R}_+^2\mid  p\lambda_b^{\frac{\alpha}{2}}\geq c_2\r\}.
\end{equation}

\item Next, consider the deployment of the hybrid network. Define the  feasibility region $\mathcal{F}_h$ for the hybrid network  as all  combinations of $(q, \lambda_p, \lambda_b)$ that satisfy the outage constraint constraint for the cellular network. Note that  $p$ is determined by $(q, \lambda_p)$. Assume that mobiles have  large  energy storage. For isotropic MPT, it is proved that 
\begin{equation}
\mathcal{F}_h = \l\{(q, \lambda_p, \lambda_b)\in\mathds{R}^3_+\mid q\lambda_p\lambda_b^{\frac{\alpha}{2}} \geq c_3\r\}. 
\end{equation}
For directed MPT, an inner bound on $\mathcal{F}_h$ is obtained that reduces to the following expression if  $\frac{z_mq}{\sigma^2}$ is sufficiently large:
\begin{equation}\label{Intro:Feasibility:b}
 \l\{(\lambda_p, \lambda_b)\in\mathds{R}^2\mid z_mq\lambda_p\lambda_b^{\frac{\alpha}{2}} \geq c_4, \l(z_m q\r)^{\frac{2}{\alpha}} \lambda_b \geq c_5\r\}\subset\mathcal{F}_h 
\end{equation}
where  $z_m$ denotes the array gain, the first condition for the inner bound ensures that PBs are sufficiently dense, and the second condition corresponds to that the maximum power transferred from a PB to an intended mobile is sufficiently high.  The results quantify the tradeoffs between $(q, \lambda_p, \lambda_b)$ and show that the gain of MPT beamforming is equivalent to increasing $q$ by a factor of $z_m$. 

\item Assume that mobiles have small energy storage. The potentially random mobile-transmission power  can be stabilized by enforcing the constraint that the received raw power at each mobile exceeds a given threshold with high probability. Under this and the outage constraint, inner bounds on $\mathcal{F}_h$ are obtained for isotropic MPT as 
\begin{equation} \label{Intro:Feasibility:c}
 \l\{(q, \lambda_p, \lambda_b)\in\mathds{R}^3_+\mid q^{\frac{2}{\beta}}\lambda_p\lambda_b^{\frac{\alpha}{\beta}} \geq c_6,  q^{\frac{2}{\alpha}}\lambda_b \geq c_7\r\}\subset\mathcal{F}_h 
\end{equation}
and for directed MPT as 
\begin{equation}\label{Intro:Feasibility:d}
 \l\{(q, \lambda_p, \lambda_b)\in\mathds{R}^3_+\mid \l(z_mq\r)^{\frac{2}{\beta}}\lambda_p\lambda_b^{\frac{\alpha}{\beta}} \geq c_8, \l(z_mq\r)^{\frac{2}{\alpha}}\lambda_b \geq c_9\r\}\subset\mathcal{F}_h. 
\end{equation}
Similar remarks as on the results in \eqref{Intro:Feasibility:b} also apply to those in \eqref{Intro:Feasibility:c} and \eqref{Intro:Feasibility:d}. Moreover, compared with their counterparts  for mobiles with large  energy storage, the current tradeoffs between $\lambda_p$ and $\lambda_b$ depend on  the path-loss exponents for both data and MPT links. In particular, with other parameters fixed, $\lambda_p$ decreases inversely with increasing $\lambda_p$ if $\alpha = \beta$.  
\end{enumerate}

The remainder of the paper is organized as follows. The models, metrics  and constraints are described in Section~\ref{Section:Model}. The feasibility regions for the cellular and hybrid networks are analyzed in Section~\ref{Section:CellularNet} and \ref{Section:HybridNet}, respectively. Numerical and simulation results are presented in Section~\ref{Section:Sim} followed by concluding remarks in Section~\ref{Section:Conclusion}. 

{\bf Notation:} The notation is summarized in Table~\ref{Table:Notation}. 

\begin{table}[t!]
\caption{Summary of notation}
\begin{center}
\begin{tabular}{cp{10cm}}
\toprule
{\bf Symbol}  &   {\bf Meaning}\\
\midrule
$\Phi$, $\lambda_b$ & PPP of base stations, density of $\Phi$\\

$\Psi$, $\lambda_p$ & PPP of PBs, density of $\Psi$\\

$Y_0, U_0, T_0$ & Typical base station, mobile and PB\\

$U_Y$ &Mobile served by base station $Y$\\

$\alpha$, $\beta$ & Path-loss exponent for a data channel, a MPT channel\\

$\Pout$, $\epsilon$ & Outage probability and its maximum for the typical base station\\

$p, q$ &Transmission power for  mobiles, PBs\\

$\theta$ & Outage threshold \\

$\mathcal{C}$ & Set of base stations serving mobiles whose interference to $Y_0$ is canceled    \\

$p_b$ & Target received signal power at a BS in an interference-limited network  \\

$\eta$ & Maximum probability that received signal power at a BS is below $p_b$  in an interference-limited network  \\
$p_t$, $\delta$ & Required received raw power at a mobile and the maximum power-outage probability  for the case of mobiles having small energy storage  \\
$\mathds{R}_+$ & Set of nonnegative numbers\\
$K$ & Number of interfering mobiles each BS cancels\\
\bottomrule
\end{tabular}
\end{center}
\label{Table:Notation}
\end{table}

\begin{figure}[t]
\begin{center}
\includegraphics[width=11cm]{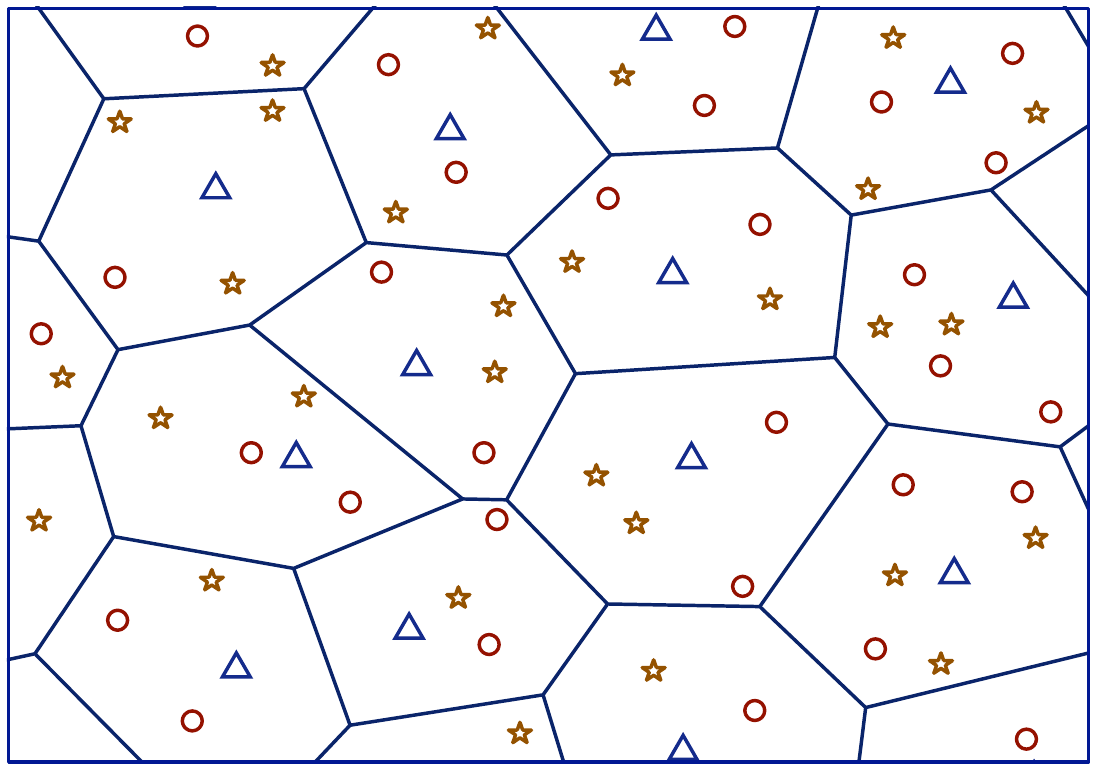}
\end{center}
\caption{Hybrid-network model where BSs, PBs and mobiles are marked with triangles, circles and stars, respectively. The Voronois cells with respect to  BSs  are plotted with solid  lines.  }
\label{Fig:NetworkModel}
\end{figure}

\section{Models, Metrics and Constraints}\label{Section:Model}

\subsection{Access-Network Model} In the uplink network, BSs are modeled as a homogeneous  PPP $\Phi = \{Y\}$ with density $\lambda_b$, where $Y\in\mathds{R}^2$ represents the coordinates of the corresponding BS. Given that mobiles are associated with their nearest BSs, the Euclidean plane is partition into Voronoi cells with points in $\Phi$ being the nuclei as illustrated in Fig.~\ref{Fig:NetworkModel}. Mobiles in the same cell are independent and uniformly distributed, and time share the corresponding BS. All active mobiles transmit signals with fixed power $p$. By applying Slyvnyak's Theorem \cite{StoyanBook:StochasticGeometry:95}, 
a typical base station, denoted as $Y_0$, is assumed to be located at the origin without loss of generality.  The active mobile served by $Y_0$ is called the typical active mobile and denoted as $U_0$. It is assumed that $Y_0$ uses a multiuser detector  \cite{VerBook} to perfectly cancel interference from $K$ nearest interfering mobiles. For ease of notation, let $U_Y$ denote the active mobile served by BS $Y$ and $\mathcal{C}$ represent the set of BSs serving $K$ nearest interfering mobiles for $Y_0$. Channels are characterized by path loss based on the long-range propagation model such that signals transmitted by a mobile $U$  with power $a$  is  received by BS $Y$ with power  $a|U - Y|^{-\alpha}$  where $\alpha > 2$ is the path-loss exponent. It follows that the signal power received at $Y_0$  is $p|U_0|^{-\alpha}$ and the interference power is given as
\begin{equation}\label{Eq:I}
I = \sum_{Y\in \Phi\backslash\{Y_0, \mathcal{C}\}} p|U_Y|^{-\alpha}. 
\end{equation}
Last, time is slotted and all channels and point processes including BSs, mobiles and PBs (discussed in the sequel)  are assumed fixed within one slot and independent over different slots.

\subsection{Power-Beacon Network Model} PBs are modeled as a homogeneous  PPP, denoted as $\Psi$, with density $\lambda_p$ and independent with the BS and mobile processes. Each mobile deploys a MP receiver with a dedicated antenna as shown in Fig.~\ref{Fig:Harvester} to intercept  microwave  energy  transmitted by PBs. Inspired by the design in \cite{Zhang:OptimalSaveThenTransmitEnergyHarvesting}, the power receiver in Fig.~\ref{Fig:Harvester} deploys two energy-storage units to enable continuous MPT. When a mobile  is active,  one unit is used  for MPT and the other for powering the transmitter; their roles are switched when the latter unit is fully discharged. When the transmitter is inactive, the MP receiver remains active till all units are fully charged. 

\begin{figure}[t]
\begin{center}
\includegraphics[width=11cm]{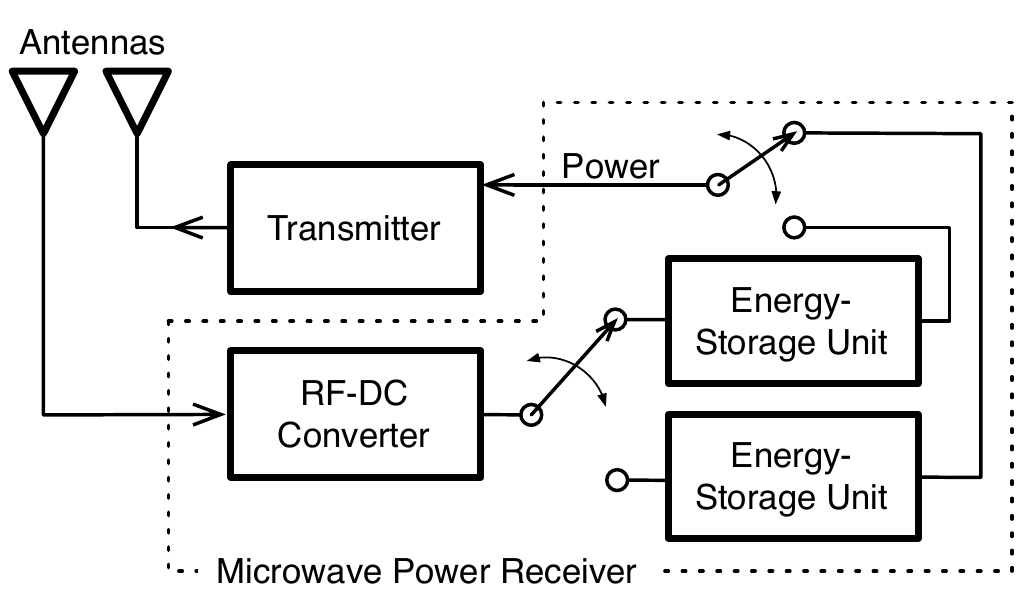}
\end{center}
\caption{Microwave-power receiver  that uses  two energy-storage  units to enable simultaneous energy harvesting and powering a transmitter.}
\label{Fig:Harvester}
\end{figure}

Consider isotropic MPT. A short-range propagation model \cite{Baccelli:AlohaProtocolMultihopMANET:2006} is used that avoids singularity caused by proximity between PBs and mobiles and thereby ensures finite average raw (unmodulated)  power received by a mobile. To be specific, 
the raw power received at the typical active mobile $U_0$ is given as
\begin{equation}\label{Eq:RXPwr:NoBeam}
P = q \sum_{T\in \Psi} \l[ \max(|T - U_0|, \nu)\r]^{-\beta}
\end{equation}
where $\nu >  1$ is a constant, the path-loss exponent $\beta > 2$ need not be equal to $\alpha$ for data links,  and the received signal power from other mobiles is neglected as it is much smaller than raw power from PBs.

Next, consider directed MPT. Each  mobile is constrained to be charged by only the nearest PB using beamforming. To simplify analysis, the beamforming responses  are assumed to have only two fixed  levels, namely the main-lobe and side-lobe levels. A PB can serve multiple mobiles simultaneously where the total transmitted raw power is multiplied to ensure the reliability of MPT. Specifically, the transmission power of $T \in \Phi$ is $M_Tq$ where $M_T$ is the random number of mobiles for which $T$ is the nearest PB.  Based on these assumptions,  the raw power received by $U_0$ can be written as 
\begin{equation}\label{Eq:RXPwr:Beam:a}
P =  z_mq \max\l(|U_0 - T_0|, \nu\r)^{-\beta} +  z_sq \sum_{T\in\Psi \backslash \{T_0\}}\max(|U_0 - T|, \nu)^{-\beta}\nn
\end{equation}
where the constants  $z_m, z_s > 0$   represent the beam main-lobe and side-lobe responses, respectively.  In contrast with the isotropic counterpart, directed MPT requires handshaking between PBs and mobiles e.g.,  using the following protocol. By estimating the relative distances and identifications of nearby PBs from their continuously broadcast signals, a mobile sends a charging request and pilot symbols to the nearest beacon that then estimate the mobile direction and beams raw power in this direction. 

\subsection{Metrics and Constraints}

The cellular-network performance is measured by the outage probability that also gives the fraction of mobiles outside network coverage \cite{Andrews:TractableApproachCoverageCellular:2010}. Define the outage probability for the typical BS $Y_0$ as 
\begin{equation}\nn
\Pout = \Pr\l(\frac{p |U_0|^{-\alpha}}{I+ \sigma^2} < \theta\r)
\end{equation}
where  $I$ is  given in \eqref{Eq:I}, $\sigma^2$ denotes   the channel-noise variance, and $\theta$ is the target SINR.  
\begin{constraint} [Cellular network] \emph{The cellular-network deployment satisfies  the outage constraint: $\Pout \leq \epsilon$ with $0 < \epsilon < 1$.}
\end{constraint}
The special case of an interference-limited network is also considered.  To this end, besides the outage constraint, the following one  is applied to ensure sufficiently high received signal power at BSs for the cellular network to operate in the interference-limited regime. 

\begin{constraint}[Interference-limited cellular  network]\emph{The received signal power at a BS is required to exceed a threshold $p_b$ with probability larger than $(1-\eta)$ with $p_b > 0$ and $0 < \eta < 1$. }
\end{constraint}

Mobiles store energy in either large storage unites (e.g., rechargeable batteries) or small ones (e.g., super capacitors). Large storage at a mobile removes the randomness of received raw power and provides  fixed output power (see Lemma~\ref{Lem:TXProb}). However, given small storage, the transmission power of a mobile  may fluctuate due to   the random locations of PBs. The fluctuation is  suppressed by ensuring sufficiently large and stable  instantaneous power transferred to each mobile as follows.

\begin{constraint} [Power beacons] \emph{Given mobiles having \emph{small energy storage}, the PB deployment satisfies the constraint that the received raw power at each  mobile is below    a given threshold $p_t > 0$ (power outage) with probability no larger than a constant $\delta$ with $0 < \delta <  1$.}
\end{constraint}
The probability $\delta$ is chosen to be sufficiently small such that a power-shortage event  at an active mobile can be coped with by drawing  energy  buffered in the small storage units to ensure uninterrupted transmission with power $p$.

\section{Access-Network Deployment} \label{Section:CellularNet}

In this section, the requirements for the cellular-network parameters $(p, \lambda_b)$ are investigated by deriving  the feasibility region $\mathcal{F}_c$. It is found that $\mathcal{F}_c$ for the case of interference-limited networks is independent with the outage constraint and is determined only by the constraint on the BS received signal power (see Constraint $3$). 

\subsection{Access-Network Deployment with Nonzero Noise}

To characterize  the feasibility region,  a useful result is obtained as shown in Lemma~\ref{Lem:Dilute} that is proved in Appendix~\ref{App:Dilute} using Mapping Theorem \cite{Kingman93:PoissonProc}. 
\begin{lemma}\label{Lem:Dilute}\emph{Consider the BS process $\Phi$ as a function of the density $\lambda_b$ denoted as $\Phi(\lambda_b)$. The two random variables  $\sum_{Y\in \Phi(\lambda_b)} h_Y|U_Y|^{-\alpha}$ and $\sum_{Y\in \Phi(a \lambda_b)} a^{-\frac{\alpha}{2}}h_Y|U_Y|^{-\alpha}$ follow an identical distribution, where $a > 0$ and $h_Y\in\mathds{R}$ depends only on $Y$.  
}
\end{lemma}
Furthermore, define $\mu > 0$ such that 
\begin{equation}\label{Eq:Mu:Def}
\Pr\l(\sum_{Y\in \Lambda(1)} \l[|U_Y|^{-\alpha} 1\{Y \neq Y_0\cup\mathcal{C}\}-\theta^{-1} |U_0|^{-\alpha}1\{Y = Y_0\}\r]+ \mu > 0\r)=\epsilon. 
\end{equation}
The main result of this section is shown  in the following proposition. 

\begin{proposition} \label{Prop:Feasibility}\emph{Under Constraint~$1$, the feasibility region for the cellular network  is given as 
\begin{equation}\label{Eq:SetF:Def}
\mathcal{F}_c = \l\{(\lambda_b, p) \in \mathds{R}_+^2\mid p \lambda_b^{\frac{\alpha}{2}} \geq \frac{\sigma^2}{\mu}\r\} 
\end{equation}
with $\mu$ defined in \eqref{Eq:Mu:Def}. 
}
\end{proposition}

\begin{proof} Consider the outage probability $\Pout(\cdot, \cdot)$ as a function of $(\lambda_b, p)$. 
To facilitate analysis, the outage probability can be rewritten as 
\begin{equation}\label{Eq:Pout:Def}
\Pout(\lambda_b, p) = \Pr\l(\sum_{Y\in \Phi(\lambda_b)} \l[|U_Y|^{-\alpha} 1\{Y \neq Y_0\cup \mathcal{C}\}-\theta^{-1} |U_0|^{-\alpha}1\{Y = Y_0\}\r]+ \frac{\sigma^2}{p} > 0\r). 
\end{equation}
Using \eqref{Eq:Pout:Def} and Lemma~\ref{Lem:Dilute}, it can be obtained that 
\begin{equation}\label{Eq:Pout:Noise}
\Pout(\lambda_b, p) = \Pr\l(\sum_{Y\in \Phi(1)} \l[|U_Y|^{-\alpha} 1\{Y \neq Y_0\cup \mathcal{C}\}-\theta^{-1} |U_0|^{-\alpha}1\{Y = Y_0\}\r]+ \frac{\sigma^2}{p\lambda_b^{\frac{\alpha}{2}}} > 0\r). 
\end{equation}
Since the summation in \eqref{Eq:Pout:Noise} is a continuous random variable, there exists $\mu > 0$ as defined in \eqref{Eq:Mu:Def} such that $\Pout(p, \lambda_b)\leq \epsilon $ for all $\sigma^2/(p\lambda_b^\alpha)\leq \mu$. The desired result follows. 
\end{proof}
Two remarks are in order.
\begin{enumerate}
\item The result in Proposition~\ref{Prop:Pout:IntLim} is consistent with the intuition that deploying denser BSs shortens transmission distances and hence requires lower transmission power at mobiles. In particular, doubling $\lambda_b$ decreases $p$  by a factor of $2^{\frac{\alpha}{2}}$, which is more significant for more severe propagation attention (larger $\alpha$). 

\item It is challenging to derive a closed-form expression for $\mu$ defined in \eqref{Eq:Mu:Def}  and its relation with $\epsilon$ is evaluated via  simulation in the sequel.  However, some simple properties of the relation  can be inferred including that a) $\epsilon(\mu)$  is a strictly monotone decreasing function of $\mu$, b) $\epsilon(0)$  corresponds to  interference limited networks and c) $\lim_{\mu\rightarrow\infty}\epsilon(\mu) = 1$.

\end{enumerate}

\subsection{Access-Network Deployment with Zero  Noise}

Consider an interference-limited cellular network where  noise is negligible.  

\begin{proposition}\label{Prop:Pout:IntLim}\emph{For an interference-limited cellular network with zero noise, $\Pout$ is independent with the BS density $\lambda_b$ and mobile-transmission power $p$. }
\end{proposition}
\begin{proof}
By substituting $\sigma^2=0$ into \eqref{Eq:Pout:Def}, the outage probability can be written as 
\begin{equation}\label{Eq:Pout:IntLim}
\Pout = \Pr\l(\sum_{Y\in \Phi} \l[|U_Y|^{-\alpha} 1\{Y \neq Y_0\cup \mathcal{C}\}-\theta^{-1} |U_0|^{-\alpha}1\{Y = Y_0\}\r] > 0\r). 
\end{equation}
Let  $\Pout$ and the BS process $\Phi$ be functions of $\lambda_b$, dented as $\Pout(\lambda_b)$ and $\Phi(\lambda_b)$, respectively. Using \eqref{Eq:Pout:IntLim} and given $a > 0$, $\Pout(a \lambda_b)$ can be written as
\begin{align}
\Pout(a\lambda_b) &= \Pr\l(\sum\nolimits_{Y\in \Phi(a\lambda_b)} \l[|U_Y|^{-\alpha} 1\{Y \neq Y_0\cup \mathcal{C}\}-\theta^{-1} |U_0|^{-\alpha}1\{Y = Y_0\}\r] > 0\r)\nn\\
&= \Pr\l(\sum\nolimits_{Y\in \Phi(\lambda_b)} \l[|U_Y|^{-\alpha} 1\{Y \neq Y_0\cup \mathcal{C}\}-\theta^{-1} |U_0|^{-\alpha}1\{Y = Y_0\}\r] > 0\r)\label{Eq:Pout:IntLim:a}\\
& = \Pout(\lambda_b)\label{Eq:Pout:IntLim:b}
\end{align}
where \eqref{Eq:Pout:IntLim:a} applies Lemma~\ref{Lem:Dilute} and \eqref{Eq:Pout:IntLim:b} follows from \eqref{Eq:Pout:IntLim}. The proposition statement  is a direct result of \eqref{Eq:Pout:IntLim:b}, completing the proof.
\end{proof}

A few remarks are made. 

\begin{enumerate}
\item The dual  of  the result in Proposition~\ref{Prop:Pout:IntLim}   for downlink networks  with fading and without interference cancelation is shown in \cite{Andrews:TractableApproachCoverageCellular:2010} using the method of Laplace transform. These results are the consequence of the fact that in both uplink and downlink interference-limited networks, transmission power has no effect on the signal-to-interference ratio (SIR), and increasing the BS density shortens the distances of data and interference links by the same factor and hence also has  no effect  on the SIR. Therefore, such results hold regardless of if fading is present and interference cancelation is applied. 

\item For interference-limited networks, the outage probability is largely affected  by the outage threshold $\theta$ and the pass-loss exponent $\alpha$ that determines the level of spatial separation. 

\item Though outage probability is independent with  $\lambda_b$ in an interference-limited network, the  network throughput grows linearly with increasing $\lambda_b$ since denser active  mobiles can be supported. 

\end{enumerate}

The  result in Proposition~\ref{Prop:Pout:IntLim} implies that Constraint~$1$ has no effect on the feasibility region for an interference-limited network. To derive the region under Constraint~$2$, the  distribution function of the received signal  power at the typical BS has no closed-form but can be upper bounded as 
\begin{equation}\label{Eq:RXPower:CDF}
\Pr\l(p|U_0|^{-\alpha} \leq p_t \r)\leq \Pr\l(p\l(\min_{Y\in\Phi\backslash\{Y_0\}} \frac{|Y|}{2}\r)^{-\alpha}\leq t\r)
\end{equation}
that results from the assignment of mobiles to the nearest BSs. Note that $\Phi\backslash\{Y_0\}$ follows the same distribution as $\Phi$ based on Slyvnyak's Theorem \cite{StoyanBook:StochasticGeometry:95}. The distribution of the distance between two nearest points in a PPP is well known \cite{Haenggi:DistUniformRandomNetwk:2005}, yielding that 
\begin{equation}\label{Eq:NearestPoints:CCDF}
\Pr\l(\min_{Y\in\Phi\backslash\{Y_0\}} |Y|\geq r\r) = e^{-\pi\lambda_b r^2}, \qquad r \geq 0. 
\end{equation}
Combining \eqref{Eq:RXPower:CDF} and \eqref{Eq:NearestPoints:CCDF} and applying Constraint~$2$ lead to the following proposition. 
\begin{proposition}\label{Prop:RXPower:Outage} \emph{Consider the deployment of an interference-limited cellular network under Constraint~$1$ and $2$. The feasibility region is independent  with Constraint~$1$ and inner bounded as 
\begin{equation}
\l\{(p, \lambda_b)\in\mathds{R}_+^2\mid p\lambda_b^{\frac{\alpha}{2}}\geq \frac{1}{\tilde{\mu}}\r\} \subset \mathcal{F}_c 
\end{equation}
with $\tilde{\mu} = \frac{1}{p_b} \l(\frac{2\pi}{\log\frac{1}{\eta}}\r)^{\frac{\alpha}{2}}$.
}
\end{proposition}
By comparing Proposition~ \ref{Prop:Feasibility} and \ref{Prop:RXPower:Outage}, the conditions  on $p$ and $\lambda_b$ for both the cases of interference-limited and nonzero-noise networks have the same form, namely $p \lambda^{\frac{\alpha}{2}} \geq c$ with $c$ being a constant, and differ only in  $c$. The similarity is a consequence of the fact that increasing $p$ and $\lambda_b$ affects only  the received signal power but not the  SIR (see Remark~$1$ on Proposition~\ref{Prop:Pout:IntLim}). 

\section{Hybrid-Network  Deployment}\label{Section:HybridNet}
In this section, the hybrid-network deployment is analyzed  by deriving the feasibility region $\mathcal{F}_h$ for different PB-network configurations combining isotropic/directed MPT and large/small energy storage at mobiles. 

\subsection{Hybrid-Network Deployment: Mobiles with Large  Energy Storage} 
In this section, it is assumed that energy-storage units (see Fig.~\ref{Fig:Harvester}) at mobiles have infinite capacity.  Large storage provides active mobiles reliable transmission power as specified in the following lemma directly following from  \cite[Theorem~1]{Huang:WirelessAdHocNetworkEnergyHarvesting} that studies energy harvesting in wireless ad hoc networks. 

\begin{lemma}[\cite{Huang:WirelessAdHocNetworkEnergyHarvesting}]\label{Lem:TXProb}\emph{For mobiles with infinite energy storage and powered by MPT, the probability that a mobile can transmit signals with power $p$  is given as
\begin{equation}
\Pr(P \geq p) = \l\{
\begin{aligned}
&1, && \E[P] \geq \omega p  \\
& \frac{\E[P]}{\omega p}, && \text{otherwise}
\end{aligned}
\r.
\end{equation}
where $\omega$ is the \emph{duty cycle} defined as the probability that a mobile transmits a packet in a slot and $P$ is the received raw power in \eqref{Eq:RXPwr:NoBeam} and \eqref{Eq:RXPwr:Beam} for isotropic and directed MPT, respectively.
}
\end{lemma}
Lemma~\ref{Lem:TXProb} suggests that an active  mobile can transmit signals continuously with power up to $\E[P]/\omega$.  
Closed-form expressions for $\E[P]$ are derived  by applying Campbell's Theorem \cite{Kingman93:PoissonProc}. The results are shown in Lemma~\ref{Lem:HarPower}.

\begin{lemma}\label{Lem:HarPower}\emph{The  average received raw power for a mobile is given as follows:
\begin{enumerate}
\item for isotropic MPT, 
\begin{equation}\label{Eq:HarPower:NoBeam}
\E[P] = \frac{ \pi\beta \nu^{2-\beta} q \lambda_p}{\beta -2 };
\end{equation}
\item for directed MPT, 
\begin{equation}\label{Eq:HarPower:Beam}
\E[P] = z_mq\psi + z_sq\l(\frac{ \pi\beta  \nu^{2-\beta}  \lambda_p}{\beta -2 }-\psi\r)
\end{equation}
where 
\begin{equation}
\psi = \nu^{-\beta} \l(1 - e^{-\pi \lambda_p \nu^2}\r) + (\pi\lambda_p)^{\frac{\beta}{2}}\gamma\l(\pi\lambda_p \nu^2, 1-\frac{\beta}{2}\r).  \nn
\end{equation} 
\end{enumerate}
}
\end{lemma}

The proof of Lemma~\ref{Lem:HarPower} is provided in Appendix~\ref{Proof:HarPower}. For a sanity check, substituting $z_m = z_s = 1$ into \eqref{Eq:HarPower:Beam} gives the result for isotropic MPT in \eqref{Eq:HarPower:NoBeam}. Two remarks are given  as follows. 

\begin{enumerate}
\item For isotropic MPT,  a mobile is powered by all nearby PBs and thus the expected received raw power $\E[P]$ given in \eqref{Eq:HarPower:NoBeam}  is proportional to the PB density $\lambda_p$. The value of $\E[P]$ is large if the path-loss exponent $\beta$ for MPT links is closed to $2$, which corresponds to  the scenario where there are LOS links between a mobile and many PBs.  In addition, $\E[P]$ is also proportional to the PB-transmission power $q$ and duty cycle $\omega$. 

\item For directed MPT, there is only one PB transferring significant raw power to each mobile. Consequently, the first term at the right-hand side of \eqref{Eq:HarPower:Beam} is dominant given that $z_m \gg z_s$ and hence $\E[P] \approx z_mq\psi$ where $\psi$ quantifies the propagation loss. Furthermore, if $\lambda_p \nu^2\ll 1$, $\psi\approx \pi\lambda_p \nu^{2-\beta}$ and hence $\E[P] \approx \pi z_mq\lambda_p \nu^{2-\beta}$ that is equal to the counterpart for isotropic MPT in \eqref{Eq:HarPower:NoBeam} with the factor $\frac{\beta}{\beta - 2}$ replaced with $z_m$. It is important to note that these two factors represent respectively the advantages  of isotropic and directed MPT, namely that strong raw power at a mobile comes  form a large number of PBs for the case of isotropic MPT with low propagation loss ($\beta$ is close to $2$) and from sharp beamforming for the case of directed MPT (large $z_m$). 
\end{enumerate}

The main result of this section is shown in Proposition~\ref{Prop:Feasibility:LargeBatt}, which follows from combining  Proposition~\ref{Prop:Feasibility}, Lemma~\ref{Lem:TXProb} and \ref{Lem:HarPower}, and the fact that the last term in \eqref{Eq:HarPower:Beam} is positive (see the proof of Lemma~\ref{Lem:HarPower}). 
 
\begin{proposition}\label{Prop:Feasibility:LargeBatt}\emph{Given   mobiles with infinite  energy storage, the feasibility region $\mathcal{F}_h$ for the hybrid network deployed under the outage constraint (Constraint~$1$)  is given as follows. 
\begin{enumerate}
\item For isotropic MPT, 
\begin{equation}\label{Eq:PBDen:NoBeam}
\mathcal{F}_h = \l\{(q, \lambda_b, \lambda_p)\in \mathds{R}^3_+\mid q\lambda_p \lambda_b^{\frac{\alpha}{2}} \geq \l(1-\frac{2}{\beta}\r)\frac{\sigma^2\omega \nu^{\beta-2}}{\pi\mu}\r\}.  
\end{equation}
\item For directed MPT, $\mathcal{F}_n$ is inner bounded as 
\begin{equation}\label{Eq:PBDen:Beam}
 \l\{(q, \lambda_b, \lambda_p)\in \mathds{R}_+^3\mid \lambda_p  \geq \frac{1}{\pi\nu^2}\log\frac{z_m q\lambda_b^{\frac{\alpha}{2}}}{z_m q\lambda_b^{\frac{\alpha}{2}}-\kappa}, \l(z_m q\r)^{\frac{2}{\alpha}}\lambda_b \geq \kappa^{\frac{2}{\alpha}}\r\}\subset \mathcal{F}_h 
\end{equation}
with $\kappa = \frac{\omega \sigma^2\nu^\beta}{\mu}$. 
\end{enumerate}
}
\end{proposition} 

Two remarks are in order. 

\begin{enumerate}
\item For isotropic MPT, the required PB density decreases with the increasing PB transmission power $q$ and BS density  $\lambda_b$ as shown in \eqref{Eq:PBDen:NoBeam}. The result also shows the effects of the path-loss exponents $\alpha$ and $\beta$ for data and MPT links, respectively. Specifically, the required PB density diminishes  as $\beta$ approaches  $2$ corresponding to fee-space propagation and increases with growing $\alpha$ for which data-link path loss is more severe and higher transmission power at mobiles is needed.  In addition, for both isotropic and directed MPT, the required PB density is proportional to the mobile duty cycle $\omega$.

\item Consider the inner bound on $\mathcal{F}_h$ in \eqref{Eq:PBDen:Beam}. The first condition ensures that PBs and their served mobiles are   sufficiently close (within a distance of $\nu$) with high probability. Under the second condition, a single PB is able to provide sufficient  power to an intended and nearby mobile by either sharp beamforming (large $z_m$) or/and high transmission power (large $q$). 
\end{enumerate}

If the ratio $z_m q/\sigma^2$ is large, the inner bound on $\mathcal{F}_h$ can be simplified as shown in the following corollary. 

\begin{corollary}\label{Cor:Feasibility:LargeBatt}\emph{Given   mobiles with infinite  energy storage and powered by directed MPT, as $\frac{z_m q}{\sigma^2}\rightarrow\infty$,  the feasibility region $\mathcal{F}_h$ for the hybrid network deployed under the outage constraint (Constraint~$1$)  can be inner bounded as 
\begin{equation}
 \l\{(q, \lambda_b, \lambda_p)\in \mathds{R}_+^3\mid z_m q\lambda_p\lambda_b^{\frac{\alpha}{2}}\geq \frac{\kappa}{\pi\nu^2} + o(1), \l(z_m q\r)^{\frac{2}{\alpha}}\lambda_b \geq \kappa^{\frac{2}{\alpha}}\r\}\subset \mathcal{F}_h.  
\end{equation}
}
\end{corollary} 

Last, the similarity between Proposition~\ref{Prop:Feasibility} and Proposition~\ref{Prop:RXPower:Outage} results in the following corollary. 
\begin{corollary}\label{Cor:Feasiblity:IntLimNet:LargeStorage}\emph{By replacing $\sigma^2/\mu$ with $\tilde{\mu}$, the results in Proposition~\ref{Prop:Feasibility:LargeBatt} and Corollary~\ref{Cor:Feasibility:LargeBatt} also apply to the case of an interference-limited cellular network with zero noise and mobiles having infinite  energy storage. 
}
\end{corollary}

\subsection{Hybrid-Network Deployment: Mobiles with Small Energy Storage} For the current case, the deployment of the hybrid network  has to satisfy both  Constraint~$1$ and $3$. Th derivation of the resultant feasibility region requires analyzing the cumulative-distribution function (CDF) of the received raw power at a mobile that is a summation over a PPP. Though the function has no closed form \cite{Lowen:PowerLawShotNoise:1990}, it can be upper bounded as shown in the following lemma by  considering  only the nearest PB for a  mobile. 

\begin{lemma}\label{Lem:EnPout}\emph{Given  mobiles with small energy storage,  the power-outage probability is upper bounded as follows: 
\begin{enumerate}
\item for isotropic MPT, if $q\nu^{-\beta} \geq p$, 
\begin{equation}
\Pr(P < p) \leq e^{-\pi \lambda_p \l(\frac{q}{p}\r)^{\frac{2}{\beta}}};  
\end{equation}
\item for directional  MPT,  if $z_mq\nu^{-\beta} \geq p$, 
\begin{equation}
\Pr(P < p) \leq e^{-\pi \lambda_p \l(\frac{z_m q}{p}\r)^{\frac{2}{\beta}}}. 
\end{equation}
\end{enumerate}
}
\end{lemma}

The proof of Lemma~\ref{Lem:EnPout} is provided in Appendix~\ref{App:EnPout}. 
Combining the results in  Proposition~\ref{Prop:Feasibility}, Lemma~ \ref{Lem:EnPout} and applying  Constraint~$3$ give the main result of this section as shown below. 

\begin{proposition}\label{Prop:Feasibility:SmallBatt} \emph{Given   mobiles with small energy storage, the feasibility region $\mathcal{F}_h$ for the hybrid network deployed under Constraint~$1$ and $3$ can be inner bounded  as follows:
\begin{enumerate}
\item for isotropic MPT, 
\begin{equation}
\l\{(\lambda_b, \lambda_p)\in \mathds{R}_+^2\mid q^{\frac{2}{\beta}}\lambda_p\lambda_b^{\frac{\alpha}{\beta}} \geq \frac{\log\frac{1}{\delta}}{\pi}\l(\frac{\sigma^2}{\mu }\r)^{\frac{2}{\beta}}, q^{\frac{2}{\alpha}}\lambda_b \geq \l(\frac{\sigma^2\nu^\beta}{\mu }\r)^{\frac{2}{\alpha}}\r\}\subset \mathcal{F}_h; 
\end{equation}
\item for directional  MPT, 
\begin{equation}
\l\{(\lambda_b, \lambda_p)\in \mathds{R}_+^2\mid (z_m q)^{\frac{2}{\beta}}\lambda_p\lambda_b^{\frac{\alpha}{\beta}} \geq \frac{\log\frac{1}{\delta}}{\pi}\l(\frac{\sigma^2}{\mu  }\r)^{\frac{2}{\beta}},  (z_mq)^{\frac{2}{\alpha}}\lambda_b \geq \l(\frac{\sigma^2\nu^\beta}{\mu }\r)^{\frac{2}{\alpha}}\r\}\subset\mathcal{F}_h. 
\end{equation}
\end{enumerate}
}
\end{proposition}

It can be observed from the results that the product of the PB parameters, namely $q^{\frac{2}{\beta}}\lambda_p$,  is proportional to the logarithm of the maximum power-shortage probability $\delta$ and hence insensitive to changes on $\delta$. Moreover, the results in Proposition~\ref{Prop:Feasibility:SmallBatt} show that beamforming reduces the required $q^{\frac{2}{\beta}}\lambda_p$ by the factor of $(z_m)^{\frac{2}{\beta}}$. Like the case of mobiles having large energy storage, the required value of $q^{\frac{2}{\beta}}\lambda_p$ for the current case also decreases with increasing BS density but at a different rate. 

Last, the results in Proposition~\ref{Prop:Feasibility:SmallBatt} have their counterparts for the interference-limited network as shown in the following corollary obtained similarly as Corollary~\ref{Prop:Feasibility:SmallBatt}. 
\begin{corollary}\emph{By replacing $\sigma^2/\mu$ with $\tilde{\mu}$, the results in Proposition~\ref{Prop:Feasibility:SmallBatt} also apply to the case of an interference-limited cellular network with zero noise and mobiles having small energy storage. 
}
\end{corollary}

\section{Numerical and Simulation Results}\label{Section:Sim}
The results in the section are obtained based on the following values of parameters unless specified otherwise. The path-loss exponents are $\alpha = 4$ and $\beta = 3$ and the outage threshold is  $\theta = 2$. Each BS cancels interference from $K=8$ nearest unintended  mobiles. Considering Constraint~$1$, the maximum outage probability for each BS is $\epsilon = 0.3$. The parameters in Constraint $2$ and $3$ are set as $p_b = 10$ dB and $\eta =\delta = 0.2$ while $p_t$ is a variable.   Last, the PB transmission power is fixed as $q = 17$ dB.\footnote{In the path-loss models, the mobile/PB transmission power is normalized by  a reference value measured at a  reference distance that is used to normalize all propagation distances  \cite{rappaport}.} 

\subsection{Access-Network Deployment}

The curve of $\epsilon$ versus $\mu$ defined in \eqref{Eq:Mu:Def}  is plotted in Fig.~\ref{Fig:PoutMu} and observed to be consistent with Remark~$2$ on Proposition~\ref{Prop:Feasibility}. In particular, $\epsilon$ approaches $1$ as $\mu$ increases and its minimum value at $\mu = 0$ gives the  outage probability ($0.06$) for an interference-limited network. In addition, it is observed that  $\mu=3.8$ corresponds  to $\epsilon = 0.3$ considered for subsequent simulations. 

\begin{figure}[t]
\begin{center}
\includegraphics[width=11cm]{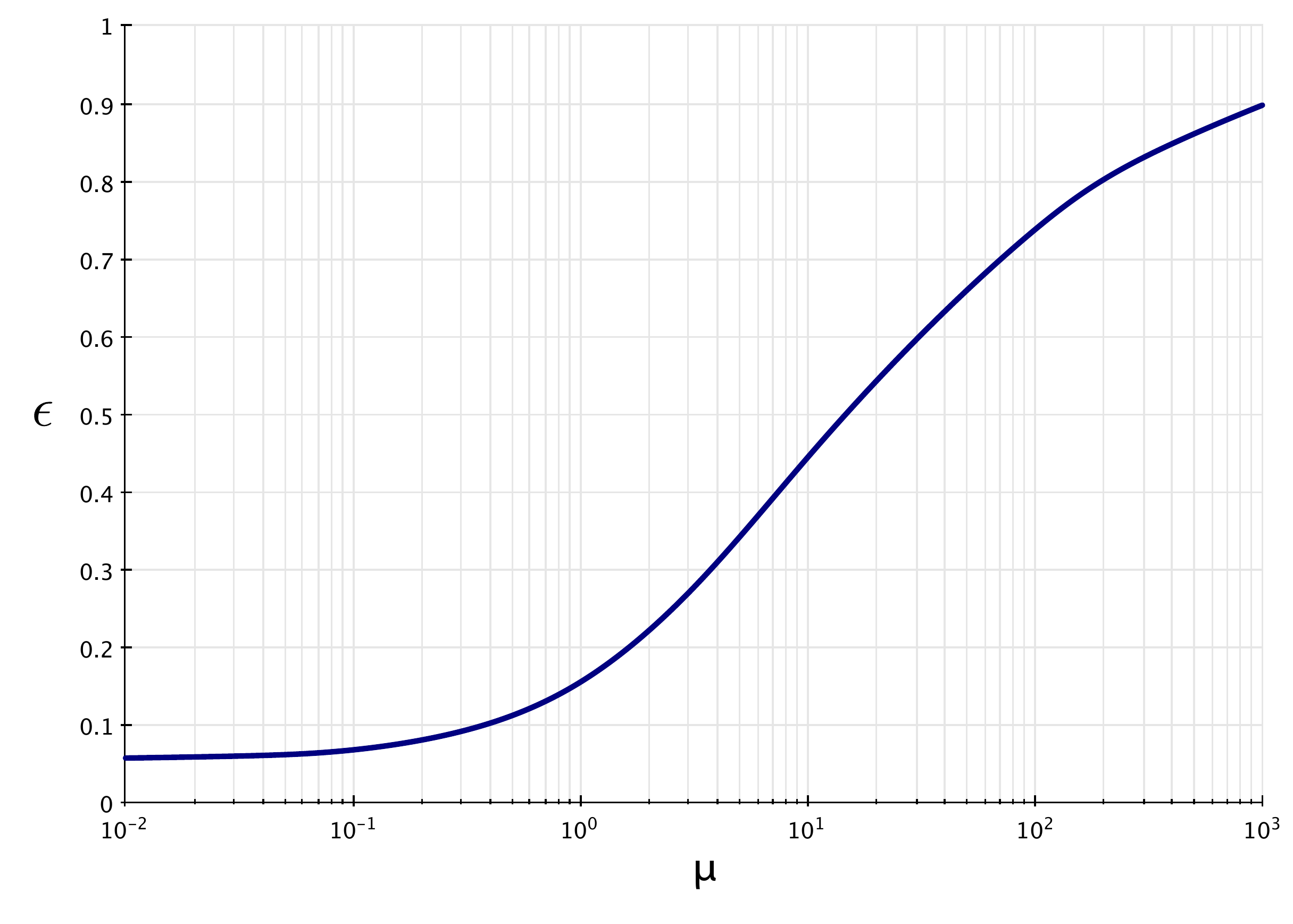}
\end{center}
\caption{$\epsilon$ versus $\mu$ defined in \eqref{Eq:Mu:Def}. }
\label{Fig:PoutMu}
\end{figure}

The feasibility regions $\mathcal{F}_c$ for the cellular network as given in Proposition~\ref{Prop:Feasibility} and \ref{Prop:Pout:IntLim} are plotted in Fig.~\ref{Fig:Feasibility:Cellular} for both the cases of non-zero and zero noise (an interference-limited network).  The boundaries of two feasibility regions are parallel as they have the same form. The minimum $p$ in dB is observed to decrease linearly with increasing $\log \lambda_b$. 

\begin{figure}[t]
\begin{center}
\includegraphics[width=11cm]{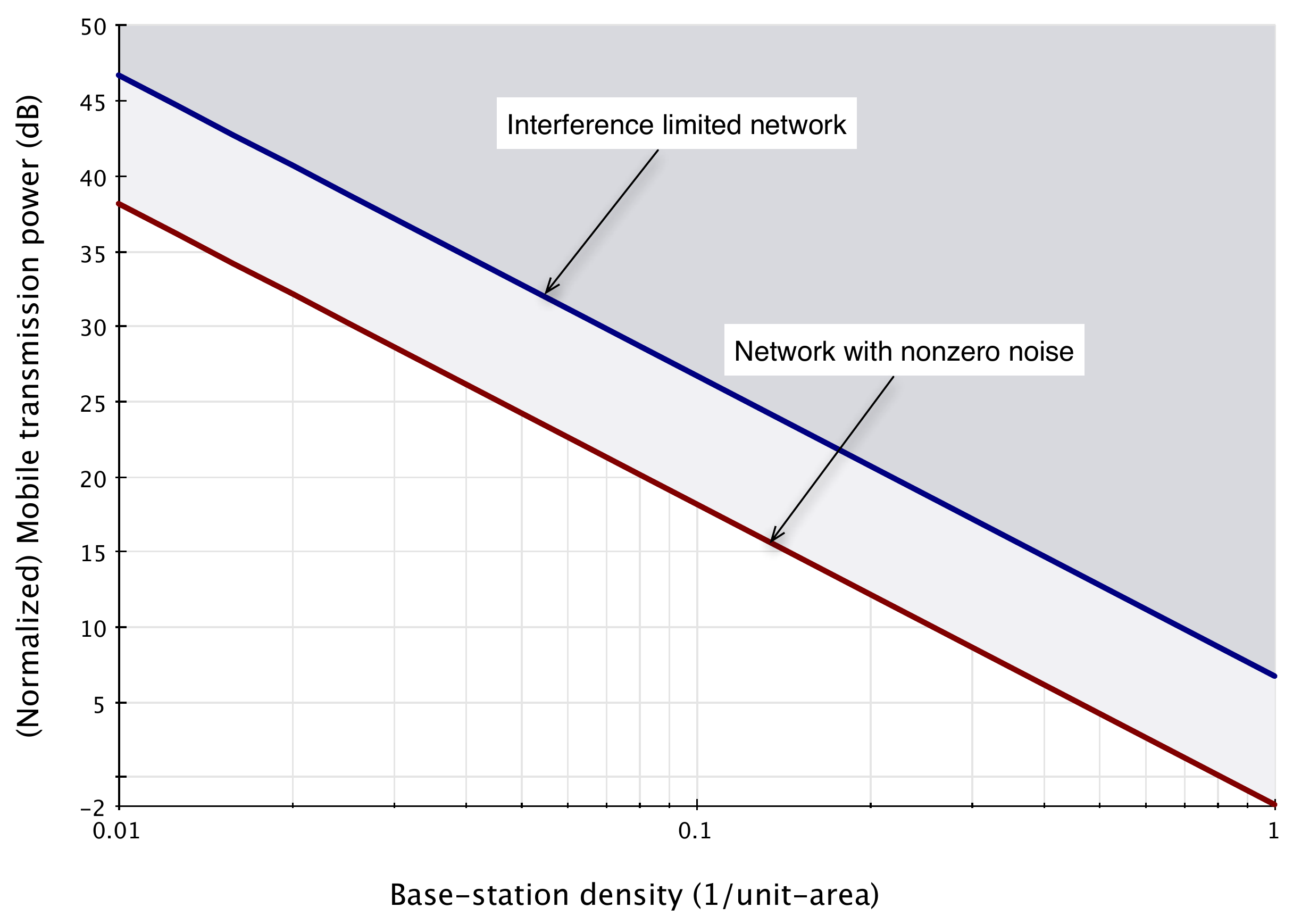}
\end{center}
\caption{Feasibility regions for the cellular network with nonzero and zero noise (an interference limited network).   }
\label{Fig:Feasibility:Cellular}
\end{figure}

\subsection{Hybrid-Network Deployment}

Fig.~\ref{Fig:TXPower} shows the curves of  mobile-transmission power $p$ versus BS density $\lambda_p$ for different combinations between large/small energy storage at mobiles and isotropic/directional MPT. The curves for the case of large energy storage are computed numerically using Lemma~\ref{Lem:HarPower} and those for the case of small storage are obtained by simulation.  On one hand, the values of $p$  for directional MPT are observed to be insensitive to changes on $\lambda_b$ if it is sufficiently large. The reason is that $p$  is dominated by the maximum raw power  transferred  to a mobile from the nearest PB that is a constant  $q\nu^{-\alpha}$. On the other hand, the values of  $p$ (in dB) for isotropic MPT do not saturate and increase approximately logarithmically with growing $\lambda_p$. The reason is that  $p$ for this  case depends on power transferred by many PBs instead of a dominant one and hence are  sensitive to changes on $\lambda_p$.

\begin{figure}[t]
\begin{center}
\includegraphics[width=11cm]{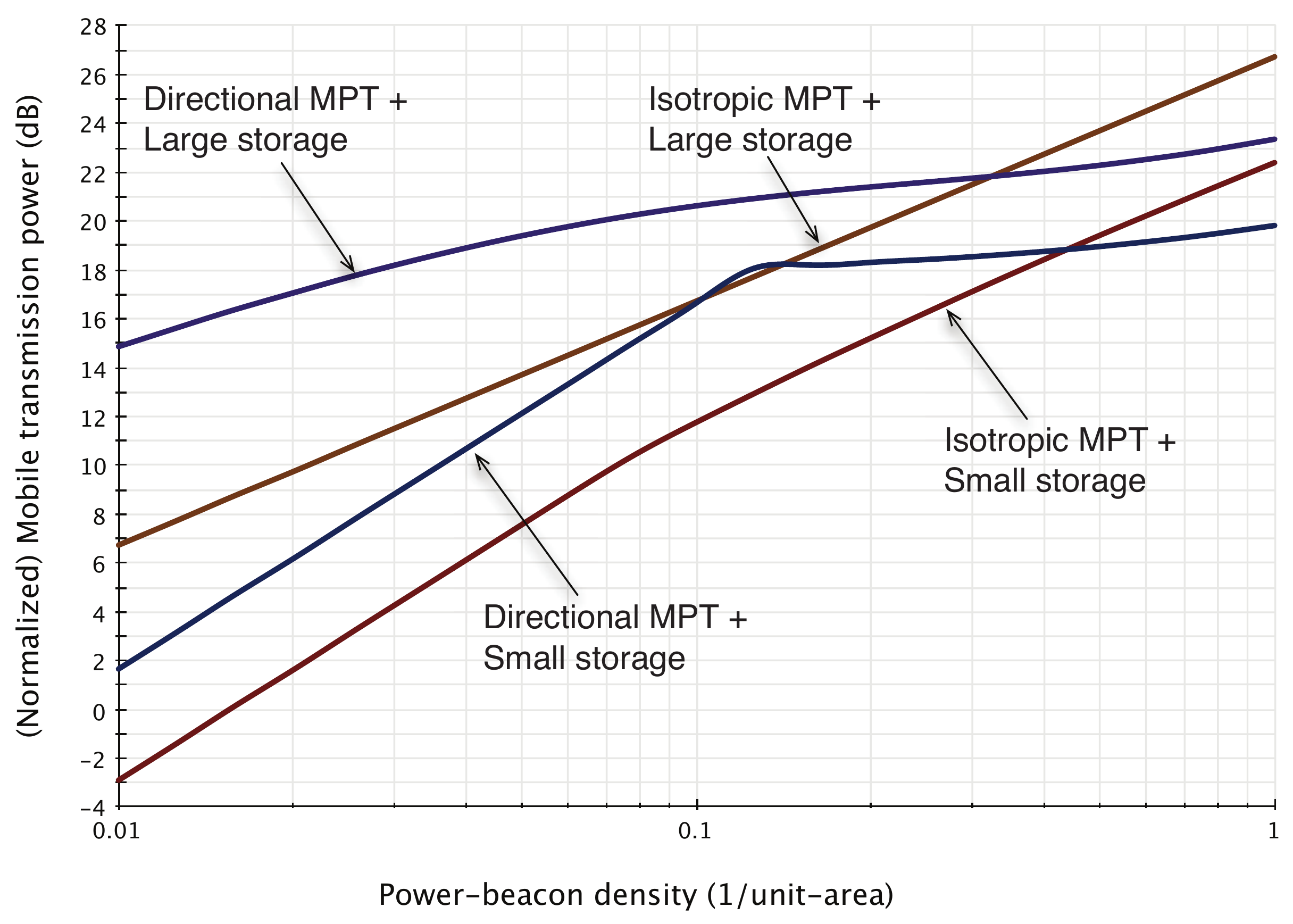}
\end{center}
\caption{Mobile-transmission power versus PB density for different combinations between large/small energy storage at mobiles and isotropic/directional MPT.}
\label{Fig:TXPower}
\end{figure}

In Fig.~\ref{Fig:Feasibility:Hybrid}, the feasibility regions for the hybrid network are plotted for different cases accounting for large/small energy storage at mobiles and isotropic/directional MPT.  The feasibility region  for the case of isotropic MPT with mobiles having large energy storage is obtained using Proposition~\ref{Prop:Feasibility} and those for other cases by simulation. Also shown  are inner bounds on the feasibility regions based on Proposition~\ref{Prop:Feasibility:LargeBatt} and \ref{Prop:Feasibility:SmallBatt} with their boundaries  plotted with  dashed lines.  The bounds are observed to be tight except for the case of isotropic MPT and mobiles with small energy storage due to the lack of a dominant PB for each mobile in terms of received raw power. By inspecting the feasibility regions for directed MPT, there exists a threshold on the BS density below which the required PB density grows extremely rapidly. The BS density  above the threshold corresponds to the regime where the power transferred to a mobile from the nearest  PB is  larger than the minimum required for satisfying the outage constraint on the cellular network. Otherwise, besides receiving power from the nearest PB,  mobiles rely on accumulating power from unintended PBs via their side-lobe transmissions so as to satisfy the outage constraint,  resulting in a high required PB density. This is the reason that isotropic MPT enlarges the feasibility region for relatively small BS densities and directed MPT is advantageous for relatively large BS densities as observed from Fig.~\ref{Fig:Feasibility:Hybrid}

\begin{figure}[t]
\begin{center}
\subfigure[Mobiles with large energy storage]{\includegraphics[width=11cm]{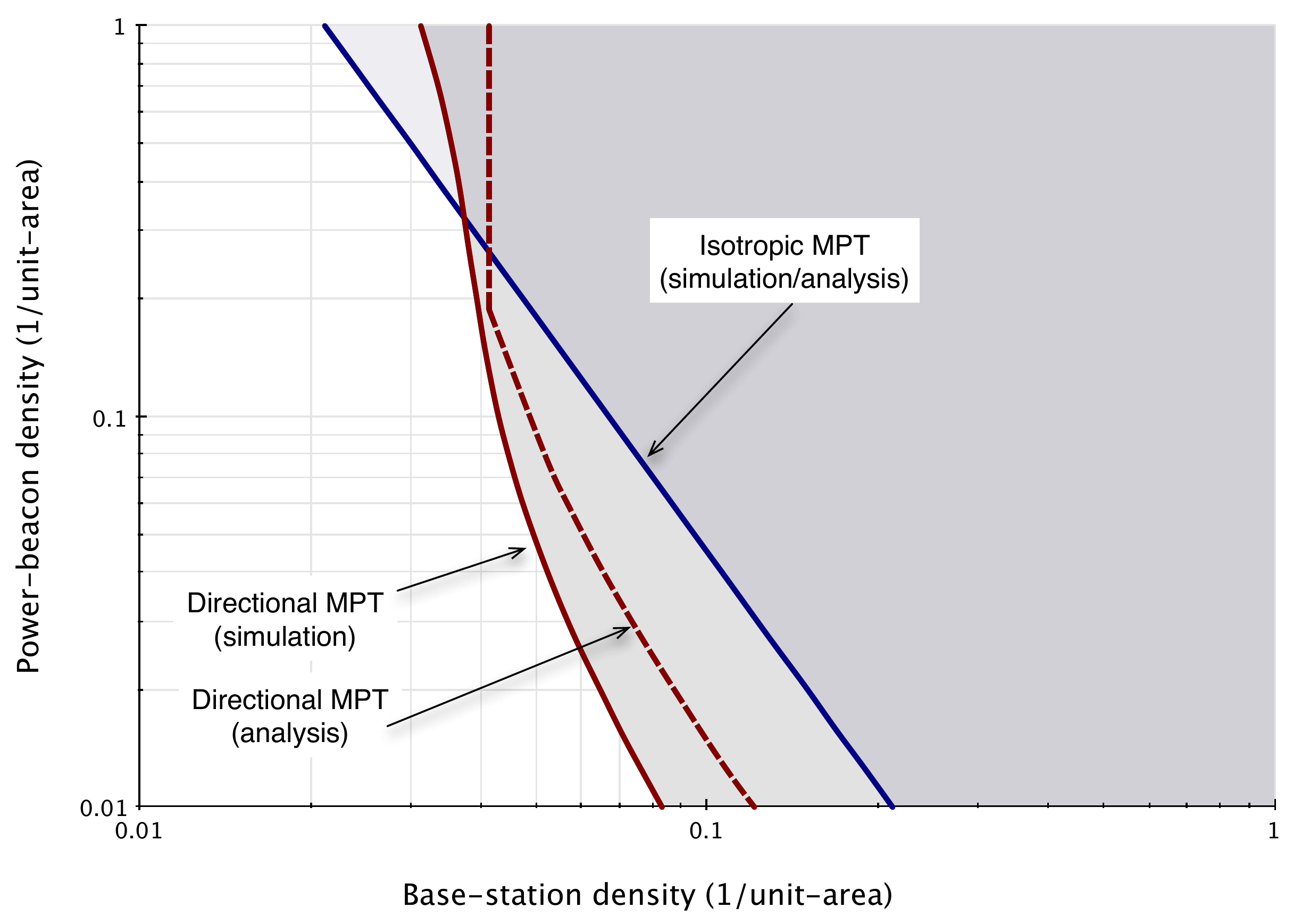}}\\
\subfigure[Mobiles with small storage]{\includegraphics[width=11cm]{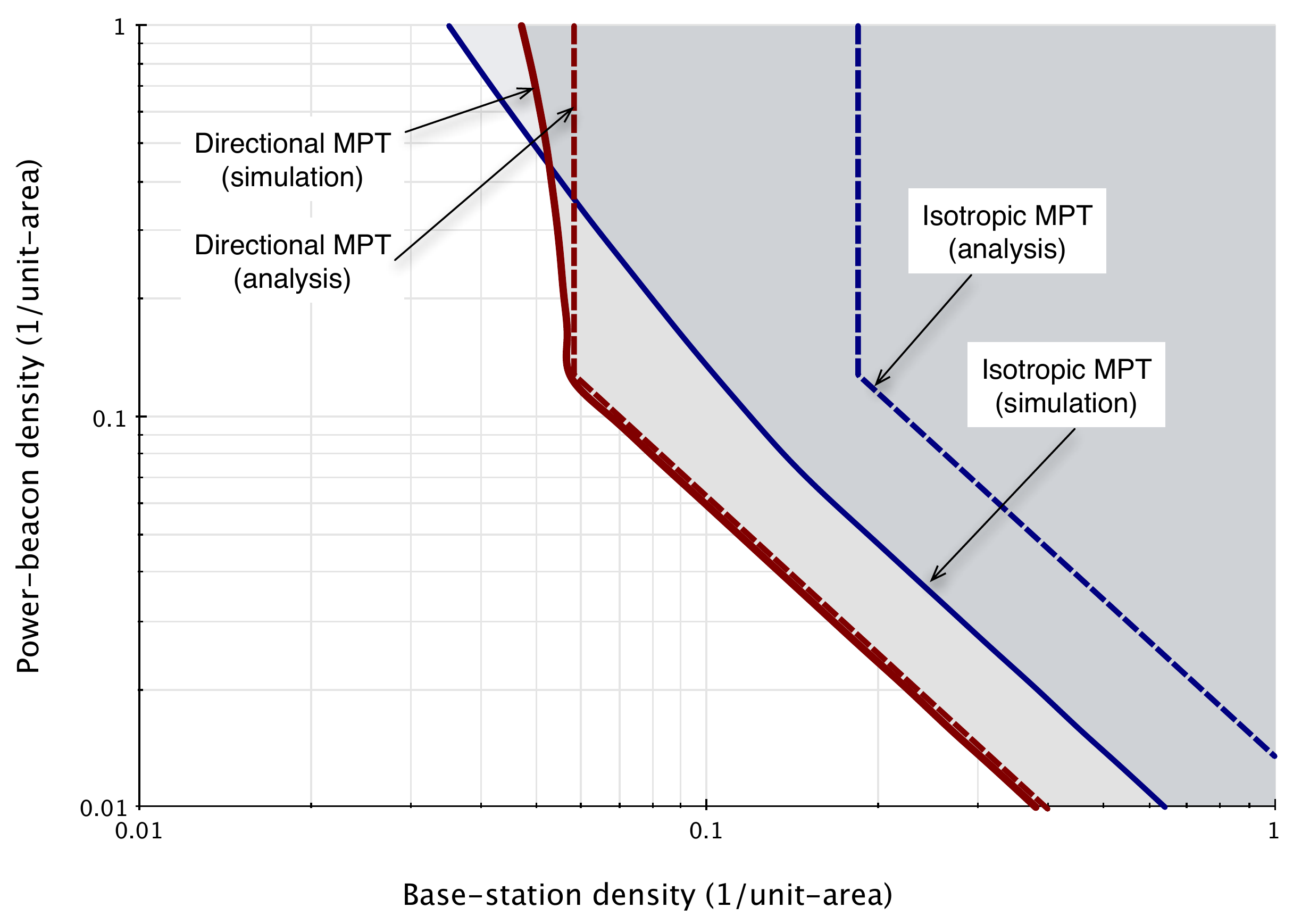}}\\
\end{center}
\caption{Feasibility regions for (a) mobiles with large storage and (b) mobiles with small storage. The boundaries of the inner bounds on the feasibility regions are plotted with dashed lines. }
\label{Fig:Feasibility:Hybrid}
\end{figure}

\section{Concluding Remarks}\label{Section:Conclusion}
The deployment of PBs for powering a cellular network  via MPT has been investigated based on a stochastic-geometry network model, resulting in simple tradeoffs between network parameters. 
First, considering the cellular network under an outage constraint, the minimum mobile-transmission power $p$ has been shown to increase  super-linearly with the decreasing base-station density $\lambda_b$. However, the absence of noise renders  outage probability independent with $(p, \lambda_b)$. Next, building on these results, the requirements on the PB deployment have been  studied by analyzing the tradeoff between $\lambda_b$ and the PB transmission power $q$ and density $\lambda_p$. Given isotropic MPT and mobiles having large capacity for storing transferred energy, the product $q\lambda_p \lambda_b^{\frac{\alpha}{2}}$ is required to be  above a given threshold so as to satisfy the outage constraint. For the case of mobiles having small energy storage, both the products  $q^{\frac{2}{\beta}}\lambda_p\lambda_b^{\frac{\alpha}{\beta}}$ and  $q^{\frac{2}{\alpha}}\lambda_b$ should be sufficient large. It has been found that compared with isotropic MPT, directed MPT by beamforming effectively increases $q$ by the array gain.

This work relies on a stochastic-geometry network model and various simplifying assumptions to derive  tradeoffs between network parameters. The results provide first-order  guidelines for realizing  MPT in cellular networks. To derive more elaborate insight into  MPT implementation, further investigations  in practical settings are necessary by considering the dependance of the directed-MPT efficiency on the array configurations, using realistic  channel models for small-cell networks, modeling the hotspot deployment of PBs using a clustered point process, characterizing the effect of mobility on MPT, and taking into account the dynamics of energy-levels at mobiles.

\renewcommand{\baselinestretch}{1.3}
\bibliographystyle{ieeetr}

\appendix 

\subsection{Proof of Lemma~\ref{Lem:Dilute}}\label{App:Dilute}
To facilitate the proof, re-denote $\{(U_Y, h_Y)\mid Y\in\Phi\}$ as $\{(U_n, h_n)\}_{n=1}^\infty$. It is almost surely that given $\Phi$ and an arbitrary active mobile $U_n$, there exists a set of coefficients $\{c_{n, Y}\mid Y\in\Phi\}$ such that $U_n = \sum_{Y\in\Phi} c_{n, Y} Y$. Consequently, a set of random coefficients $\{C_{n, Y}\mid Y\in\Phi\}$ can be defined under the constraint that $U_n = \sum_{Y\in\Phi} C_{n, Y} Y$ is uniformly distributed in the corresponding Voronoi cell. It follows that 
\begin{equation}\label{Eq:ShotNoise}
\sum\nolimits_{Y\in\Phi(\lambda_b)} h_Y |U_Y|^{-\alpha} =   \sum\nolimits_{n=1}^\infty h_n \l|\sum\nolimits_{Y\in\Phi(\lambda_b)} C_{n, Y} Y\r|^{-\alpha}. 
\end{equation}
Given $a > 0$ and $\Phi(\lambda_b)$, define a function $f: \mathds{R}^2\rightarrow\mathds{R}^2$ such that for $\mathcal{A}\subset \mathds{R}^2$, $f(\mathcal{A}) = \{\sqrt{a} x\mid x \in \mathcal{A}\}$. Intuitively, $f$ expands (or shrink) the Euclidean plane by a factor $a$ if $a > 1$ (or $ a < 1$). Let $\mu$ denote the mean measure of $\Phi(a\lambda_b)$ and define a measure $\mu^*(\mathcal{B})$ with $\mathcal{B}\subset\mathds{R}^2$ as $\mu^*(\mathcal{B}) = \mu(f^{-1}(\mathcal{B}))$. It follows that $\mu^*(\mathcal{B}) = \lambda_b |\mathcal{B}|$ where $|\mathcal{B}|$ denotes the measure (area) of $\mathcal{B}$. Since $f$ is a measurable function that does not maps distinct points to a single point, applying Mapping Theorem \cite[p18]{Kingman93:PoissonProc} gives that $f(\Phi(a\lambda_b))$ is a homogeneous PPP with density $\lambda_b$. As a result, the random variable $\sum\nolimits_{Y\in\Phi(a\lambda_b)} C_{n, Y} f(Y)$  is identically distributed as $\sum\nolimits_{Y\in\Phi(\lambda_b)} C_{n, Y} Y$. Combining this fact and  \eqref{Eq:ShotNoise} yields the desired result. \hfill $\blacksquare$

\subsection{Proof of Lemma~\ref{Lem:HarPower}}\label{Proof:HarPower}

For isotropic MPT, using \eqref{Eq:RXPwr:NoBeam} and applying Campbell's Theorem give  
\begin{align}
\E[P] &= q \lambda_p \int_{x\in\mathds{R}^2} \l[ \max(|x - U^*|, \nu)\r]^{-\alpha} dx\nn\\
& = 2q \pi \lambda_p \int_0^\infty \l[ \max(r, \nu)\r]^{-\alpha} dr\label{Eq:HarPwer:a}\\
& = 2q \pi \lambda_p \nu^{-\alpha}\int_0^\nu r dr  + 2q\pi \lambda_p \int_{\nu}^\infty r^{1-\alpha} dr \label{Eq:HarPwer:b}
\end{align}
where \eqref{Eq:HarPwer:a} results from the stationarity of $\Psi$ and using the polar-coordinate system. 
The corresponding expression of $\E[P]$ in \eqref{Eq:HarPower:NoBeam} follows from \eqref{Eq:HarPwer:b}. 

Next, consider  directed MPT and the corresponding  expression for  $\E[P]$  is derived as follows. To facilitate analysis, $P$ is rewritten from \eqref{Eq:RXPwr:Beam} as 
\begin{equation}
P =  q (z_m - z_s) \max\l(|U_0 - T_0|, \nu\r)^{-\beta} +  q z_s \sum_{T\in\Psi}\max(|U_0 - T|, \nu)^{-\beta}. \label{Eq:RXPwr:Beam}
\end{equation}
For ease of notation, define $D = |T_0 - U_0|$. Using \eqref{Eq:RXPwr:Beam} and using the result in  \eqref{Eq:HarPower:NoBeam}, 
\begin{equation}\label{Eq:HarPower:c}
\E[P] = q(z_m - z_s) \E\l[\max(D, \nu)^{-\beta}\r] +\frac{ qz_s \beta \pi \lambda_p \nu^{2-\beta}}{\beta -2 }.  
\end{equation}
 Note that  $D$ measures the shortest distance between a point in  the PPP $\Psi$ to a fixed point and have the following probability-density function  \cite{Haenggi:DistUniformRandomNetwk:2005}
\begin{equation}\label{Eq:D:PDF}
f_{D}(r) = 2\pi \lambda_p r e^{- \pi \lambda_p r^2}, \qquad r \geq 0. 
\end{equation}
Given the PDF, it is obtained that 
\begin{align}
\E\l[\max(D, \nu)^{-\beta}\r] &= 2\pi \lambda_p \int_0^\infty \max(r, \nu)^{-\beta} r e^{-\pi \lambda_p r^2} dr  \nn\\
&= 2\pi \lambda_p \int_0^\nu \nu^{-\beta} r e^{-\pi \lambda_p r^2}dr + 2\pi \lambda_p \int_\nu^\infty r^{1-\beta} e^{-\pi \lambda_p r^2} dr\nn\\
&= \nu^{-\beta} \l(1 - e^{-\pi \lambda_p \nu^2}\r) + (\pi\lambda_p)^{\frac{\beta}{2}}\gamma\l(\pi\lambda_p \nu^2, 1-\frac{\beta}{2}\r). \label{Eq:CondExp:b}
\end{align}
Substituting \eqref{Eq:CondExp:b} into \eqref{Eq:HarPower:c} gives the desired result in 
\eqref{Eq:HarPower:Beam}, completing the proof. \hfill $\blacksquare$

\subsection{Proof of Lemma~\ref{Lem:EnPout}}\label{App:EnPout}

Without loss of generality, assume $U_0$ is located at the origin $o$ since according to Slyvnyak's Theorem, 
\begin{equation}
\Pr(P < p) = \Pr(P < p\mid U_0 = o) 
\end{equation}
where $P$ is given in \eqref{Eq:RXPwr:NoBeam} and \eqref{Eq:RXPwr:Beam} for isotropic and directed MPT, respectively. The previous assumption that $Y_0$ is  at the origin is unnecessary for the proof. 

Consider isotropic MPT. Define $r_0 = \l(q/p\r)^{\frac{1}{\beta}}$. Note that any PB with a distance $r_0$ from $U_0$ can supply received raw power higher than $p$. Furthermore, $r_0 \geq \nu$ as a result of the assumption $q\nu^{-\beta} \geq p$.  Inspired by the approach in \cite{WeberAndrews:TransCapWlssAdHocNetwkSIC:2005} that studies outage probability for mobile ad hoc networks, the probability of no power shortage  is expanded as 
\begin{equation}\label{Eq:Weber}
\Pr(P \geq p) = \Pr(\Phi \cap B(0, r_0) \neq \emptyset) + \Pr(P \geq p \mid \Phi \cap B(0, r_0) = \emptyset)\Pr(\Phi \cap B(0, r_0) = \emptyset) 
\end{equation}
where $B(a, b)$ is a disk centered at $a\in \mathds{R}^2$ and with a radius $b \geq 0$. 
It follows that the power-shortage probability can be lower bounded as
\begin{align}
\Pr(P \geq p) &\geq \Pr(\Phi \cap B(0, r_0) \neq \emptyset)\nn\\
& = 1 - e^{-\pi\lambda_p r_0^2}.  \label{Eq:Weber:LB}
\end{align}
The derived result in $1)$ in the lemma statement follows. 

The result for directed MPT can be obtained similarly. This completes the proof. \hfill $\blacksquare$

\end{document}

%% file: header.tex
\newtheorem{theorem}{Theorem}

\newtheorem{constraint}[theorem]{Constraint}

\newtheorem{lemma}{Lemma}
\newtheorem{corollary}{Corollary}

\newtheorem{proposition}{Proposition}

\def\E{\mathsf{E}}

\def\phi{\varphi}

\def\l{\left}
\def\r{\right}
\def\({\left(}
\def\){\right)}

\def\b0{{\mathbf{0}}}

\newcommand{\Pout}{P_{\mathsf{out}}}

\newcommand{\nn}{\nonumber}